\documentclass[12pt,thmsa,onecolumn,laterpaper]{article}
\usepackage{amssymb}

\usepackage{amsmath}
\usepackage[all]{xy}
\usepackage{graphicx}
\usepackage[latin1]{inputenc}

\newenvironment{proof}{\noindent{\bf{Proof.}}}{\hfill$\Box$\newline}

\newtheorem{corollary}{Corollary}[section]

\newtheorem{lemma}{Lemma}[section]
\newtheorem{proposition}{Proposition}[section]
\newtheorem{theorem}{Theorem}[section]

\begin{document}

\title{\textsc{Symmetry Group of Ordered Hamming Block Space}}
\date{}
\author{Luciano Panek\thanks{%
Centro de Engenharias e Ci\^{e}ncias Exatas, UNIOESTE, Av. Tarqu\'{\i}nio
Joslin dos Santos, 1300, CEP 85870-650, Foz do Igua\c{c}u, PR, Brazil.
Email: lucpanek@gmail.com} \and Nayene Michele Pai\~{a}o Panek\thanks{%
Centro de Engenharias e Ci\^{e}ncias Exatas, UNIOESTE, Av. Tarqu\'{\i}nio
Joslin dos Santos, 1300, CEP 85870-650, Foz do Igua\c{c}u, PR, Brazil.
Email: nayene@gmail.com}}
\maketitle

\begin{abstract}
Let $P = (\{1,2,\ldots,n,\leq)$ be a poset that is an union of disjoint chains of the same length and $V=\mathbb{F}_q^N$ be the space of $N$-tuples over the finite field $\mathbb{F}_q$. Let $V_i = \mathbb{F}_q^{k_i}$, $1 \leq i \leq n$, be a family of finite-dimensional linear spaces such that $k_1+k_2+\ldots +k_n = N$ and let $V = V_1 \oplus V_2 \oplus \ldots \oplus V_n$ endow with the poset block metric $d_{(P,\pi)}$ induced by the poset $P$ and the partition $\pi=(k_1,k_2,\ldots,k_n)$, encompassing both Niederreiter-Rosenbloom-Tsfasman metric and error-block metric. In this paper, we give a complete description of group of symmetries of the metric space $(V,d_{(P,\pi)})$, called the ordered Hammming block space. In particular, we reobtain the group of symmetries of the Niederreiter-Rosenbloom-Tsfasman space and obtain the group of symmetries of the error-block metric space.

\vspace{0.5cm}

\end{abstract}

%\pagewiselinenumbers

%\begin{linenumbers}

\section{Introduction}

One of the main classical problem of the coding theory is to find sets with $
q^{k}$ elements in $\mathbb{F}_{q}^{n}$, the space of $n$-tuples over the
finite field $\mathbb{F}_{q}$, with the largest minimum distance possible.
There are many possible metrics that can be defined in $\mathbb{F}_{q}^{n}$,
the most common ones are the Hamming and Lee metrics.

In 1987 Harald Niederreiter generalized the classical problem of coding
theory (see \cite{HN1}). Brualdi, Graves and Lawrence (see \cite{Bru}) also
provided in 1995 a wider situation for the above problem: using partially
ordered sets (\textit{posets}) and defining the concept of poset codes, they started to study codes with a poset metric. This has been a fruitful approach, since many new perfect codes have been found with such poset metrics \cite{HK, L}. Later Feng, Xu and Hickernell (\cite{FXH}, 2006) introduced the block metric, by partitioning the set of coordinate positions of $\mathbb{F}_{q}^{n}$ into families of blocks. Both kinds of metrics are generalizations of the Hamming metric, in the sense that the latter is attained when considering the trivial order (in the poset case) or one-dimensional blocks (in the block metric case). In 2008, Alves, Panek and Firer (see \cite{Alves}) combined the poset and block structure, obtaining a further generalization, the poset block metrics.

A particular instance of poset block codes and spaces, with one-dimensional blocks, are the spaces introduced by
Niederreiter in 1991 (see \cite{HN1}) and Rosenbloom and Tsfasman in 1997 (see \cite{RT}), where the posets taken into
consideration have a finite number of disjoint chains of equal size, that
is, it is isomorphic to the order $P=P_{1}\cup P_{2}\cup\ldots\cup P_{m}$
such that
\[
P_{i+1}=\left\{ in+1,in+2,\ldots,\left( i+1\right) n\right\}
\]
and
\[
in+1<in+2<\ldots<\left( i+1\right) n
\]
are the only strict comparabilities for each $i\in\left\{ 0,1,\ldots
,m-1\right\} $. This spaces are of special interest since there are several rather disparate applications, as noted by Rosenbloom and Tsfasman (see \cite{RT}) and Park e Barg (see \cite{PB}).

The description of linear symmetries of a poset space started with the study
of particular poset spaces (as Lee's work on Niederreiter-Rosenbloom-Tsfasman spaces \cite
{Lee}; Cho and Kim's work on crown spaces \cite{CH}; Kim's work on weak
spaces \cite{S}), until Panek, Firer, Kim and Hyun \cite{PFKH} gave a full
description of the group of linear symmetries of a poset space. The full description of the group of linear symmetries of a poset block space were determined by Alves, Panek and Firer \cite{Alves}. The description of symmetries (not necessarily linear ones) of a poset space were studied by Panek, Alves and Firer (in the case of a product of Rosenbloom-Tsfasman space, see \cite{Alves2}) and by Hyun (to any poset, see \cite{H}). In this
work, we describe the symmetry group (not necessarily linear ones) of the
poset block space that is a finite union of disjoint chains of same length, the Niederreiter-Rosenbloom-Tsfasman block space. We call this space the \textit{ordered Hamming block space}.

In the section 2, we introduce briefly the main concepts and definitions
used in this work. In the section 3, we study the simple, but inspiring,
case of posets determining a single chain (Theorem \ref{prop3}) and finally,
in the last two sections, we describe the symmetry group of ordered Hamming block space (Theorem \ref{Theorem RT}).

\section{Poset Block Metric Space}

Let $\left[ n\right] :=\{1,2,\ldots,n\}$ be a finite set with $n$ elements
and let $\leq$ be a partial order on $\left[ n\right] $. We call the pair $%
P:=(\left[ n\right] ,\leq)$ a \emph{poset} and say that $k$ is \emph{smaller
than} $j$ if $k\leq j$ and $k\neq j$. An \emph{ideal} in $(\left[ n\right]
,\leq)$ is a subset $I\subseteq\left[ n\right] $ that contains every element
that is smaller than some of its elements, i.e., if $j\in I$ and $k\leq j$
then $k\in I$. Given a subset $X\subseteq [n]$, we denote by $\langle
X\rangle$ the smallest ideal containing $X$, called the \emph{ideal
generated by} $X$. An order on the finite set $[n]$ is called a \emph{linear
order} or a \emph{chain} if every two elements are comparable, that is,
given $i,j\in[n]$ we have that either $i\leq j$ or $j\leq i$. In this case, $%
n$ is said to be \emph{the length} of the chain and the set can be labeled
in such a way that $i_{1}<i_{2}<\ldots<i_{n}$. For the simplicity of the
notation, in this situation we will always assume that the order $P$ is
defined as $1<2<\ldots<n$.

Let $q$ be a power of a prime, $\mathbb{F}_{q}$ be the finite field of $q$
elements and $V:=\mathbb{F}_{q}^{N}$ the $N$-dimensional vector space of $N$%
-tuples over $\mathbb{F}_{q}$. Let $\pi = (k_1,k_2,\ldots,k_n)$ be a partition of $N$: $$N=k_1+k_2+\ldots+k_n$$ with $k_i >0$ a integer. For each integer $k_i$, let $V_i:=\mathbb{F}_q^{k_i}$ be the $k_i$-dimensional vector space over the finite field $\mathbb{F}_q$ and define $$V=V_1 \oplus V_2 \oplus \ldots \oplus V_n,$$ called the \textit{$\pi$-direct sum decomposition} of $V$. A vector $v \in V$ can be uniquely decomposed as $$v=v_1+v_2+ \ldots +v_n$$ with $v_i \in V_i$ for each $1 \leq i \leq n$. We will call this the \textit{$\pi$-direct sum decomposition} of $v$. Given a poset $P=([n],\leq)$, we define the \textit{poset block weight} $\omega_{(P,\pi)}$ (or simply the \textit{$(P,\pi)$-weight}) of a vector $v=v_1+v_2+ \ldots +v_n$ to be $$\omega_{(P,\pi)}(v):=\mid \langle supp(u) \rangle \mid $$ where $supp(v):=\{i \in [n]:v_i \neq 0\}$ is the \emph{$\pi$-support} of the
vector $v$ and $|X|$ is the cardinality of the set $X$. The block structure is said to be \textit{trivial} when $k_i=1$ for all $1\leq i\leq n$. The $(P,\pi)$-weight induces a metric $d_{(P,\pi)}$ on $V$, that we call the \textit{poset block metric} (or simply \textit{$(P,\pi)$-metric}): $$ d_{(P,\pi)}(u,v) := \omega_{(P,\pi)}(u-v). $$ The pair $(V,d_{(P,\pi)})$ is a metric space and where no ambiguity may rise, we say it is a \emph{poset block space}, or simply a \textit{$(P,\pi)$-space}.

A \emph{symmetry of} $(V,d_{(P,\pi)})$ is a bijection $T:V\longrightarrow V$ that
preserves distance:
\[
d_{(P,\pi)}(T(u),T(v))=d_{(P,\pi)}(u,v)
\]
for all $u,v\in V$. The set $Symm(V,d_{(P,\pi)})$ of all symmetries of $(V,d_{(P,\pi)})$ is a group with the natural
operation of composition of functions, and we call it the \emph{symmetry group} of $(V,d_{(P,\pi)})$. An \textit{automorphism} is a linear symmetry.

The description of the full symmetry group may be of help in the study of
non-linear codes. Besides other applications, linear symmetries are used to
divide linear codes in equivalence classes, since they take subspace into
subspace and preserve dimension and minimum distance. Symmetries, in
general, may take linear codes onto non-linear ones, but preserve all metric
data, such as minimal distance and weight of a code and also the generalized
Wei weights, capacity of error correction and number of elements. So it is
just natural to call two non-linear codes \emph{equivalent} if one is the
image of the other under a symmetry.

In \cite{Alves2} the group of symmetries of a product of Niederreiter-Rosenbloom-Tsfasman spaces is characterized. In \cite{H} is studied a subgroup of the full symmetry group for any given poset. In this work we will describe the full
symmetry group of an important class of poset block spaces, namely, those induced
by posets  that are an union of disjoint chains of the same length. This class
includes the block metric spaces over chains and the Niederreiter-Rosembloom-Tsfasman spaces with trivial block structures.

We remark that initial idea is the same as in \cite{Alves2}. The main differences are that we follow a more coordinate free approach an that the dimensions of the blocks pose a new restraint. We first study the symmetry group of ordered Hamming space induced for one simple chain (Theorem \ref{prop3}), analogous to those of \cite{Alves2}. Next we prove some results on symmetries, also anologous to those of \cite{Alves2}, plus a result on preservation of block dimensions (Lemma \ref{cadeia}), and conclude that $Symm(V,d_{(P,\pi)})$ is the semi-direct product of the direct product of the symmetry groups inducedes for each chain and the automorphism group of the permutations of chains that preserves the block dimensions (Theorem \ref{Theorem RT}).

\section{Symmetries of a Linear Ordered Block Space}

Let $P=([n],\leq)$ be the linear order $1<2< \ldots <n$, let $\pi = (k_1,k_2, \ldots ,k_n)$ be a partition of $N$ and let $$V= V_1 \oplus V_2 \oplus \ldots \oplus V_n,$$ $V_i = \mathbb{F}_{q}^{k_i}$, $i=1,2,\ldots,n$, be the $\pi$-direct sum decomposition of the vector space $V = \mathbb{F}_{q}^{N}$ endow with the poset block metric $d_{(P,\pi)}$. In this section we will describe the
full symmetry group of the poset block space $(V,d_{(P,\pi)})$. This description will be used in
the next section to describe the symmetry group of the ordered Hamming
block space. In this section $P=([n],\leq)$ will be total order $1<2< \ldots <n$.

We note that, given $u=(u_{1},\ldots,u_{n})$ and $v=(v_{1},\ldots,v_{n})$ in
the total ordered block space $V$,
\[
d_{(P,\pi)}(u,v)=\max\{i:u_{i}\neq v_{i}\}.
\]
For each $i\in\{1,2,\ldots,n\}$, let $$F_{i}:V_i \oplus V_{i+1} \oplus \ldots \oplus V_n \rightarrow V_i$$ be a map that is a bijection with respect to the first block space $V_i$, that is, given $v_{i+1},\ldots,v_{n}\in V_{i+1} \oplus \ldots \oplus V_n$, the map $%
\widetilde{F}_{v_{i+1},\ldots,v_{n}}:V_i\rightarrow V_i
$ defined by
\[
\widetilde{F}_{v_{i+1},\ldots,v_{n}}\left( v_{i}\right) =F_{i}\left(
v_{i},v_{i+1},\ldots,v_{n}\right)
\]
is a bijection. Given such a family, we define a map $T_{(F_{1},F_{2},%
\ldots,F_{n})}:V\rightarrow V$ by
\[
T_{(F_{1},F_{2},\ldots,F_{n})}(v_{1},\ldots,v_{n}):=(F_{1}(v_{1},\ldots
,v_{n}),F_{2}(v_{2},\ldots,v_{n}),\ldots,F_{n}(v_{n})).
\]

\begin{lemma}
\label{lema1} Let $P=([n],\leq )$ be the linear order $1<2<\ldots <n$ and
let $V = V_1 \oplus V_2 \oplus \ldots \oplus V_n$ be the $\pi$-direct sum decomposition of $V=\mathbb{F}_{q}^N$ endowed with the poset block metric induced by the
poset $P$ and the partition $\pi$. The map $T_{(F_{1},F_{2},\ldots ,F_{n})}:V\longrightarrow V$ is a
symmetry of $V$.
\end{lemma}

\begin{proof}
Given $u=(u_{1},\ldots ,u_{n}),v=(v_{1},\ldots ,v_{n})\in V$, let $%
l=d_{(P,\pi)}(u,v)=\max \{i:u_{i}\neq v_{i}\}$. Since each $F_{i}:V_i \oplus V_{i+1} \oplus \ldots \oplus V_n \rightarrow V_i$ is a bijection in relation to the
first coordinate, we have that
\[
F_{l}(u_{l},u_{l+1},\ldots ,u_{n})\neq F_{l}(v_{l},v_{l+1},\ldots ,v_{n})
\]
and
\[
F_{t}(u_{t},u_{t+1},\ldots ,u_{n})=F_{t}(v_{t},v_{t+1},\ldots ,v_{n})
\]
for any $t>l$. It follows that
\begin{align*}
d_{(P,\pi)}\left( T_{(F_{1},\ldots ,F_{n})}(u),T_{(F_{1},\ldots ,F_{n})}(v)\right)
& =\max \{i:F_{i}(u_{i},\ldots ,u_{n})\neq F_{i}(v_{i},\ldots ,v_{n})\} \\
& =l
\end{align*}
and hence $T_{(F_{1},F_{2},\ldots ,F_{n})}$ is distance preserving. Since $V$
is a finite metric space, it follows that $T_{(F_{1},F_{2},\ldots ,F_{n})}$
is also a bijection.
\end{proof}

In the previous lemma we attained a large set of symmetries of $V$. The
following lemma shows every symmetry may be expressed in this form.

\begin{lemma}
\label{lemasimetrias} Let $P=([n],\leq )$ be the linear order $1<2<\ldots <n$ and
let $V = V_1 \oplus V_2 \oplus \ldots \oplus V_n$ be the $\pi$-direct sum decomposition of $V=\mathbb{F}_{q}^N$ endowed with the poset block metric induced by the
poset $P$ and the partition $\pi$. Let $T:V\rightarrow V$ be a symmetry of $V$. Then, there are functions $F_{i}:V_i \oplus V_{i+1} \oplus \ldots \oplus V_n \rightarrow V_i$ such that:

\begin{enumerate}
\item[$\left( i\right) $]  $T\left( v_{1},v_{2},\ldots ,v_{n}\right) =\left(
F_{1}\left( v_{1},v_{2},\ldots ,v_{n}\right) ,F_{2}\left( v_{2},\ldots
,v_{n}\right) ,\ldots ,F_{n}\left( v_{n}\right) \right) $;

\item[$\left( ii\right) $]  For every $i\in \left\{ 1,\ldots ,n\right\} $
and each $\left( v_{i+1},\ldots ,v_{n}\right) \in V_{i+1} \oplus \ldots \oplus V_{n}$ the function $\widetilde{F}_{v_{i+1},\ldots ,v_{n}}:V_i \rightarrow V_i$ defined by $\widetilde{F}_{v_{i+1},\ldots
,v_{n}}\left( v_{i}\right) =F_{i}\left( v_{i},v_{i+1},\ldots ,v_{n}\right) $
is a bijection.
\end{enumerate}
\end{lemma}

\begin{proof}
Let us write
\[
T\left( v_{1},v_{2},\ldots ,v_{n}\right) =\left( T_{1}\left(
v_{1},v_{2},\ldots ,v_{n}\right) ,\ldots ,T_{n}\left( v_{1},v_{2},\ldots
,v_{n}\right) \right) \text{.}
\]
We prove first that $T_{j}\left( v_{1},v_{2},\ldots ,v_{n}\right)
=F_{j}\left( v_{j},v_{j+1},\ldots ,v_{n}\right) $, that is,
$T_{j}$ does not depends on the first $j-1$ coordinates. In other
words, we want to prove that
\[
T_{j}\left( v_{1},\ldots ,v_{j-1},v_{j},\ldots ,v_{n}\right) =T_{j}\left(
u_{1},\ldots ,u_{j-1},v_{j},\ldots ,v_{n}\right)
\]
regardless of the values of the first $j-1$ coordinates. But
\begin{align*}
d_{(P,\pi)}\left( \left( u_{1},\ldots ,u_{j-1},v_{j},\ldots ,v_{n}\right) ,\left(
v_{1},\ldots ,v_{j-1},v_{j},\ldots ,v_{n}\right) \right) & =\max_{i}\left\{
i:v_{i}\neq u_{i}\right\} \\ & \leq j-1
\end{align*}
and since $T$ is a symmetry, we find that
\[
\begin{array}{c}
d_{(P,\pi)}\left( T\left( v_{1},\ldots ,v_{j-1},v_{j},\ldots ,v_{n}\right) ,T\left(
u_{1},\ldots ,u_{j-1},v_{j},\ldots ,v_{n}\right) \right) = \\
=d_{(P,\pi)}\left( \left( u_{1},\ldots ,u_{j-1},v_{j},\ldots ,v_{n}\right) ,\left(
v_{1},\ldots ,v_{j-1},v_{j},\ldots ,v_{n}\right) \right) \leq j-1
\end{array}
\]
and so
\[
T_{j}\left( v_{1},\ldots ,v_{j-1},v_{j},\ldots ,v_{n}\right) =T_{j}\left(
u_{1},\ldots ,u_{j-1},v_{j},\ldots ,v_{n}\right)
\]
for any $(v_{1},\ldots ,v_{j-1}),(u_{1},\ldots ,u_{j-1})\in V_1 \oplus \ldots \oplus V_{j-1}$ and $(v_{j},\ldots ,v_{n}) \in V_j \oplus \ldots \oplus V_n$. We find that
\[
T\left( v_{1},v_{2},\ldots ,v_{n}\right) =\left( F_{1}\left(
v_{1},v_{2},\ldots ,v_{n}\right) ,F_{2}\left( v_{2},\ldots ,v_{n}\right)
,\ldots ,F_{n}\left( v_{n}\right) \right)
\]
and the first statement is proved.

Now we need to prove that each $\widetilde{F}_{v_{i+1},\ldots ,v_{n}}$ is a
bijection, what is equivalent to prove those maps are injective. Suppose $%
\widetilde{F}_{v_{i+1},\ldots ,v_{n}}$ is not injective, so there are $%
v_{i}\neq u_{i\text{ }}$ in $V_i$ such that
\[
\widetilde{F}_{v_{i+1},\ldots ,v_{n}}\left( v_{i}\right) =\widetilde{F}%
_{v_{i+1},\ldots ,v_{n}}\left( u_{i}\right) .
\]
Considering $i$ minimal with this property, we would have
\begin{align*}
i\ \ =& \ \ d_{(P,\pi)}\left( \left( v_{1},\ldots ,v_{i},\ldots ,v_{n}\right)
,\left( v_{1},\ldots ,u_{i},\ldots ,v_{n}\right) \right)  \\
=& \ \ d_{(P,\pi)}\left( T\left( v_{1},\ldots ,v_{i},\ldots ,v_{n}\right) ,T\left(
v_{1},\ldots ,u_{i},\ldots ,v_{n}\right) \right)  \\
<& \ \ i
\end{align*}
contradicting the assumption that $T$ is a symmetry of $V$.
\end{proof}

The next theorem follows straightforward from the previous lemmas.

\begin{theorem}
\label{prop3}Let $P=([n],\leq )$ be the linear order $1<2<\ldots <n$ and
let $V = V_1 \oplus V_2 \oplus \ldots \oplus V_n$ be the $\pi$-direct sum decomposition of $V=\mathbb{F}_{q}^N$ endowed with the poset block metric induced by the
poset $P$ and the partition $\pi$. Then, the group $\textrm{Symm}(V,d_{(P,\pi)})$ of symmetries of $V$ is the set of all
maps $T_{(F_{1},F_{2},\ldots ,F_{n})}:V\rightarrow V$.
\end{theorem}

We recall that in Lemma \ref{lemasimetrias}, we have that $\widetilde
{F}_{v_{2},\ldots,v_{n}}\left( v_{1}\right) =F_{1}\left(
v_{1},v_{2},\ldots,v_{n}\right) $ is a bijection, hence a permutation of $V_1$ for each $\left( v_{2},\ldots,v_{n}\right) \in V_2 \oplus \ldots \oplus V_n$. If $S_{m}$ denotes the symmetric group of permutations of a set
with $m$ elements, since $V = \mathbb{F}_{q}^{N}$ has $q^{N}$ elements, if $V=V_1 \oplus V_2 \oplus \ldots \oplus V_n$ is the $\pi$-direct sum decomposition of $V$ with $\pi = (k_1,k_2,\ldots,k_n)$,  we can
identify the group of functions $F:V_1 \oplus V_2 \oplus \ldots \oplus V_n \rightarrow V_1$ such that $\widetilde{F}_{v_{2},\ldots,v_{n}}$ is a permutation of $V_1 = \mathbb{F}_q^{k_1}$ with the direct product $\left( S_{q^{k_1}}\right) ^{q^{N-k_1}}$. With
this notations we have the following result:

\begin{corollary}
\label{corolario} Let $P=([n],\leq )$ be the linear order $1<2<\ldots <n$ and
let $V = V_1 \oplus V_2 \oplus \ldots \oplus V_n$ be the $\pi$-direct sum decomposition of $V=\mathbb{F}_{q}^N$ endowed with the poset block metric induced by the
poset $P$ and the partition $\pi$. If $\pi = (k_1,k_2,\ldots,k_n)$, then the group of simmetries $\textit{Symm}(V,d_{(P,\pi)})$ has a semi-direct product structure $$
(S_{q^{k_1}})^{q^{N-k_1}}\mathbb{o} \left( \ldots \left((S_{q^{k_{n-1}}})^{q^{N-k_1-k_2- \ldots - k_{n-1}}} \mathbb{o} (S_{q^{k_n}})^{q^{N-k_1-k_2- \ldots -k_{n-1} - k_n}} \right) \ldots \right).$$
\end{corollary}

\begin{proof}
Let $G_{(\widehat{P},\widehat{\pi})}$ be the symmetry group $\textit{Symm}(\widehat{V},d_{(\widehat{P},\widehat{\pi})})$ of $$\widehat{V} = \widehat{V}_1 \oplus \widehat{V}_2 \oplus \ldots \oplus \widehat{V}_{n-1}$$ where $\widehat{V}_i=V_{i+1}$ for each $i=1,2,\ldots,n-1$, $\widehat{P}=([n-1],\leq )$ is the linear order $1<2<\ldots<n-1$ and $\widehat{\pi}=(k_2,k_3,\ldots,k_n)$. We claim that $$\textit{Symm}(V,d_{(P,\pi)}) = (S_{q^{k_1}})^{q^{N-k_1}}\mathbb{o} G_{(\widehat{P},\widehat{\pi})}.$$ In order to simplify notation, we will denote the elements of $\left(
S_{q^{k_1}}\right)^{q^{N-k_1}}$ by $(\pi_X)$: $$(\pi_X) := (\pi_X)_{X \in \mathbb{F}_q^{N-k_1}}.$$

The group $G_{(\widehat{P},\widehat{\pi})}$ acts on $V = V_1 \oplus \widehat{V}$ by
\[
T\left( x_{1},\ldots,x_{n}\right) =\left( x_{1},T\left(
x_{2},\ldots,x_{n}\right) \right)
\]
and $(S_{q^{k_1}})^{q^{N-k_1}}$ acts by
\[
(\pi_{X})\left( x_{1},\ldots,x_{n}\right)
=\left( \pi_{\left( x_{2},\ldots,x_{n}\right) }\left( x_{1}\right)
,x_{2},\ldots,x_{n}\right).
\]

Both groups act as groups of symmetries and both act faithfully. Therefore
these actions establish isomorphisms of these groups with subgroups of $\textit{Symm}(V,d_{(P,\pi)})$, and we identify $(S_{q^{k_1}})^{q^{N-k_1}}$ and $G_{(\widehat{P},\widehat{\pi})}$ with
their counterparts $H \cong (S_{q^{k_1}})^{q^{N-k_1}}$ and $K\cong G_{(\widehat{P},\widehat{\pi})}$
in $\textit{Symm}(V,d_{(P,\pi)})$. From the actions defined above, it is easy to see that
\[
H=\{T\in G_{n+1};T=\left( F_{1},id_{V_2},\ldots,id_{V_n}\right) \}
\]
and
\[
K=\{T\in G_{n+1};T=\left( id_{V_1},F_{2},F_{3},\ldots,F_{n+1}\right) \}
\]
where each $F_{i}$ satisfies the conditions of Lemma \ref{lema1} and $id_{V_i}$ is the identity over the vector subspace $V_i$. Clearly, $%
\textit{Symm}(V,d_{(P,\pi)})=HK$, because each symmetry of $V$ is a
composition $T_{1}\circ T_{2}$ with $T_{1}\in H$ and $T_{2}\in K$. We claim
that $\textit{Symm}(V,d_{(P,\pi)})$ is a semi-direct product of $H$ by $K$.

Let $L\in H\cap K$. Since $L\in H$, $L(x_{1},x_{2},\ldots,x_{n+1})=(x_{1}^{%
\prime},x_{2},\ldots,x_{n+1})$ and, since $L$ is also in $K$, $%
x_{1}^{\prime}=x_{1}$. Hence $L=id_V$ and the groups $H$ and $K$ intersect
trivially.

We prove now that $H$ is a normal subgroup of $\textit{Symm}(V,d_{(P,\pi)})$. In fact, since $%
\textit{Symm}(V,d_{(P,\pi)})=HK$, it suffices to check that $THT^{-1}\subset H$ for each $T\in K$%
. Let $\left( \pi_{X}\right) \in H$ and $%
T\in K$. If $\left( x_{1},\ldots,x_{n}\right) \in V$,
then
\begin{align*}
\left( T\circ(\pi_{X})\circ T^{-1}\right)\left( x_{1},\ldots,x_{n+1}\right) 
& =  \left( T\circ(\pi_{X})\right) \left(x_{1},T^{-1}\left( x_{2}, \ldots,x_{n+1}\right) \right) \\
& = T\left( \pi_{T^{-1}\left( x_{2},\ldots,x_{n+1}\right) }\left(x_{1}\right) ,T^{-1}\left( x_{2},\ldots,x_{n+1}\right) \right) \\
& = \left( \pi_{T^{-1}\left( x_{2},\ldots,x_{n+1}\right) }\left(
x_{1}\right) ,T\circ T^{-1}\left( x_{2},\ldots,x_{n+1}\right) \right) \\
& = \left( \pi_{T^{-1}\left( x_{2},\ldots,x_{n+1}\right) }\left(
x_{1}\right) ,x_{2},\ldots,x_{n+1}\right) \\
& = (\pi_{T^{-1}(X)})\left(
x_{1},x_{2},\ldots,x_{n+1}\right).
\end{align*}

This shows that $H$ is a normal subgroup of $\textit{Symm}(V,d_{(P,\pi)})$ and that $\textit{Symm}(V,d_{(P,\pi)}) =
H\rtimes K$. Using the aforementioned isomorphisms involving $H$ and $K$ we
conclude that $\textit{Symm}(V,d_{(P,\pi)}) \cong (S_{q^{k_1}})^{q^{N-k_1}}\rtimes G_{(\widehat{P},\widehat{\pi})}$.
\end{proof}

\begin{corollary}
Let $P=([n],\leq )$ be the linear order $1<2<\ldots <n$ and
let $V = V_1 \oplus V_2 \oplus \ldots \oplus V_n$ be the $\pi$-direct sum decomposition of $V=\mathbb{F}_{q}^N$ endowed with the poset block metric induced by the
poset $P$ and the partition $\pi = (k_1,k_2,\ldots,k_n)$. Then $$\vert \textit{Symm}(V,d_{(P,\pi)}) \vert = (q^{k_1}!)^{q^{N-k_1}} \cdot (q^{k_2}!)^{q^{N-k_1-k_2}} \cdot \ldots \cdot (q^{k_{n-1}}!)^{q} \cdot (q^{k_{n}}!).$$
\end{corollary}

Now, if the partition $\pi = (1,1,\ldots,1)$ we have that (see \cite{Alves2}, Corollary 3.1):

\begin{corollary}
Let $P=([n],\leq )$ be the linear order $1<2<\ldots <n$ and let $V=\mathbb{F}_{q}^n$ be the vector space endowed with the poset metric induced by the
poset $P$ (with de block structure trivial). Then the group of simmetries $\textit{Symm}(V,d_{P})$ is isomorphic to the semi-direct product $$
(S_{q})^{q^{n-1}} \mathbb{o} ( \ldots ((S_{q})^{q} \mathbb{o} S_{q}) \ldots ).$$ In particular, $$\vert \textit{Symm}(V,d_{P}) \vert = (q!)^{\frac{q^n-1}{q-1}+1}.$$
\end{corollary}

\section{Symmetries of Ordered Hamming Block Spa\-ces}

In this section we consider an order $P=\left( [m\cdot n],\leq\right) $ that
is the union of $m$ disjoint chains $P_{1},P_{2},\ldots,P_{m}$ of order $n$.
We identify the elements of $[m\cdot n]$ with the set of ordered pairs of
integers $(i,j)$, with $1\leq i\leq m,\ 1\leq j\leq n$, where $(i,j)\leq(k,l)
$ \emph{iff} $i=k$ and $j\leq_{\mathbb{N}}l$, where $\leq_{\mathbb{N}}$ is
just the usual order on $\mathbb{N}$. We denote $P_{i}=\{(i,j):1\leq j\leq
n\}$. Each $P_{i}$ is a chain and those are the connected components of $%
\left( [m\cdot n],\leq\right) $.

Let $\pi = (k_{11},\ldots,k_{1n},\ldots,k_{m1},\ldots,k_{mn})$ be a partition of $N$ and for each $1 \leq i \leq m$ let $\pi_i = (k_{i1},\ldots,k_{in})$. Given a finite field $\mathbb{F}_{q}$ and $V = U_1 \oplus U_2 \oplus \ldots \oplus U_m$, where $U_i:=V_{i1} \oplus V_{i2} \oplus \ldots \oplus V_{in}$ and $\dim(V_{ij})=k_{ij}$ for all $1\leq i\leq m,\ 1\leq j\leq n$, we identify $V$ with the set of matrices $$\left\{ \left( \begin{array}{ccc}
v_{11} & \ldots & v_{1n} \\ 
\vdots & \ddots & \vdots \\ 
v_{m1} & \ldots & v_{mn}
\end{array} \right):v_{ij} \in V_i,1\leq i\leq m,\ 1\leq j\leq n \right\}.$$
The space $V$ with the poset metric induced by
the order $P=\left( [m\cdot n],\leq\right)$ is called the \emph{$(m,n,\pi)$-ordered Hamming block space}. Note that if $n=1$, then $%
P=\left( [m\cdot1],\leq\right) $ induces just the error-block metric on $V$, and in particular, if $\pi = (1,1,\ldots, 1)$, then $%
P=\left( [m\cdot1],\leq\right) $ induces just the Hamming metric on $%
\mathbb{F}_{q}^{m}$. Hence the induced metric from the poset $P=\left(
[m\cdot n],\leq\right) $ can be viewed as a generalization of the error-block metric.

Let $V=U_{1}\oplus U_{2}\oplus\ldots\oplus U_{m}$ as above, called the
\emph{canonical decomposition of} $V$. Given the canonical decompositions $%
u=u_{1}+\ldots+u_{m}$ and $v=v_{1}+\ldots+v_{m}$ with $u_{i},v_{i}\in U_{i}$%
, it is well known that
\[
d_{(P,\pi)}(u,v)=\sum_{i=1}^{m}d_{(P_{i},\pi_{i})}(u_{i},v_{i})
\]
where $d_{(P_{i},\pi_{i})}$, the restriction of $d_{(P,\pi)}$ to $U_{i}$, is a linear poset block
metric. We note that the restriction of $d_{(P,\pi)}$ to each $U_{i}$ turns it
into a poset space defined by a linear order, that is, each $U_{i}$ is symmetric to $(U_i,d_{([n],\pi_i)})$ with the metric $d_{([n],\pi_i)}$
determined by the chain $1<2<\ldots<n$. Let $G_{i,\pi_{i},n}$ be the group of
symmetries of $(U_i,d_{([n],\pi_i)})$. The direct product $\prod_{i=1}^m G_{i,\pi_{i},n}$ acts on $V$ in the following manner: given $%
T=(T_{1},\ldots,T_{m}) \in \prod_{i=1}^m G_{i,\pi_{i},n}$ and $v\in V$,
\[
T\left( v\right) :=\sum_{i=1}^{m}T_{i}\left( v_{i}\right).
\]

\begin{lemma}
\label{m} Let $(V,d_{(P,\pi)})$ be the $(m,n,\pi)$-ordered Hamming block space over $%
\mathbb{F}_{q}$ and let $G_{i,\pi_{i},n}$ be the group of symmetries of $\left( U_i,d_{([n],\pi_i)}\right) $. Given $T_{i}\in G_{i,\pi_{i},n}$, with $1\leq
i\leq m$, the map $T=(T_{1},\ldots,T_{m})$ defined by
\[
T(v)=\sum_{i=1}^{m}T_{i}(v_{i})\label{3}
\]
is a symmetry of $(V,d_{(P,\pi)})$.
\end{lemma}

\begin{proof}
Given $u,v\in V$, consider the canonical decompositions $u=u_{1}+\ldots+u_{m}$ and $%
v=v_{1}+\ldots+v_{m}$ with $u_{i},v_{i}\in U_{i}$. Then,
\begin{align*}
d_{(P,\pi)}\left( T(u),T(v)\right) & =d_{(P,\pi)}\left(
\sum_{i=1}^{m}T_{i}(u_{i}),\sum_{i=1}^{m}T_{i}(v_{i})\right) \\
& =\sum_{i=1}^{m}d_{(P_{i},\pi_{i})}\left( T_{i}(u_{i}),T_{i}(v_{i})\right) \\
& =\sum_{i=1}^{m}d_{(P_{i},\pi_{i})}\left( u_{i},v_{i}\right) \\
& =d_{(P,\pi)}\left( u,v\right).
\end{align*}
\end{proof}

Let $S_m$ be the permutation group of $\{1,2,\ldots,m\}$. We will call a permutation $\sigma \in S_m$ \textit{admissible} if $\sigma(i) = j$ implies that $k_{il}=k_{jl}$ for all $1 \leq l \leq n$. Cleary, the set $S_\pi$ of all admissible permutations is a subgroup of $S_m$.

Let us consider the canonical decomposition $v=v_{1}+v_{2}+\ldots+v_{m}$ of
a vector $v$ in the $(m,n,\pi)$-ordered Hamming block space $V$.
The group $S_{\pi}$ acts on $V$ as a group of symmetries: given $%
\sigma\in S_{\pi}$ and $v=v_{1}+v_{2}+\ldots+v_{m}\in V$, we define
\[
T_{\sigma}(v)=v_{\sigma(1)}+v_{\sigma(2)}+\ldots+v_{\sigma(m)}
\]

\begin{lemma}
\label{n} Let $(V,d_{(P,\pi)})$ be the $(m,n,\pi)$-ordered Hamming block space $V$ and let $\sigma\in S_{\pi}$. Then $T_{\sigma}$ is a symmetry of $(V,d_{(P,\pi)}).$
\end{lemma}

\begin{proof}
Given $u,v\in V$, we consider their canonical decompositions $%
u=u_{1}+\ldots+u_{m}$ and $v=v_{1}+\ldots+v_{m}$ with $u_{i},v_{i}\in U_{i}$%
. Then,

\begin{align*}
d_{(P,\pi)}\left( T_{\sigma}(u),T_{\sigma}(v)\right) & =d_{(P,\pi)}\left( \sum
_{i=1}^{m}u_{\sigma(i)},\sum_{i=1}^{m}v_{\sigma(i)}\right) \\
& =\sum_{i=1}^{m}d_{(P_{i},\pi_{i})}\left( u_{\sigma(i)},v_{\sigma(i)}\right) \\
& =\sum_{i=1}^{m}d_{(P_{i},\pi_{i})}\left( u_{i},v_{i}\right) \\
& =d_{(P,\pi)}\left( u,v\right).
\end{align*}
\end{proof}

The two previous lemmas assure that the groups $\prod_{i=1}^m G_{i,\pi_{i},n}$ and $S_{\pi}$ are
both symmetry groups of the $(m,n,\pi)$-ordered Hamming block space $V$, and so is the group $G_{(m,n,\pi)}$ generated by both of
them. We identify $\prod_{i=1}^m G_{i,\pi_{i},n}$ and $S_{\pi}$ with their images in $G_{(m,n,\pi)}$ and make an abuse of notation, denoting the images in $G_{(m,n,\pi)}$ by
the same symbols. With this notation, analogous calculations as those of
Corollary \ref{corolario} show that $$\left(\prod_{i=1}^m G_{i,\pi_{i},n}\right) \cap S_{\pi}=\{id_{V}\}$$ and $$\sigma\circ \left(\prod_{i=1}^m G_{i,\pi_{i},n}\right) \circ \sigma^{-1}=\prod_{i=1}^m G_{i,\pi_{i},n}$$ for every $\sigma\in
S_{\pi}$. Since $\prod_{i=1}^m G_{i,\pi_{i},n}$ is normal in $G_{(m,n,\pi)}$ and $G_{(m,n,\pi)}$ is generated by $\prod_{i=1}^m G_{i,\pi_{i},n}$ and $S_{\pi}$, we have that $$G_{(m,n,\pi)} = \left(\prod_{i=1}^m G_{i,\pi_{i},n}\right) \cdot S_{\pi},$$  and we have proved the following:

\begin{proposition}
\label{Gmn} The group $G_{(m,n,\pi)}$ has the structure of a semi-direct product $$G_{(m,n,\pi)}=\left( \prod_{i=1}^m G_{i,\pi_{i},n} \right)
\mathbb{o} S_{\pi}.$$
\end{proposition}

We need two more lemmas in order to prove that every symmetry of the $(m,n,\pi)$-ordered Hamming block space $V$ is in $G_{(m,n,\pi)}$, i.e., that $G_{(m,n,\pi)}$ is the group of symmetries of $V$.

\begin{lemma}
\label{cadeia}Let $(V,d_{(P,\pi)})$ be the $(m,n,\pi)$-ordered Hamming block space and let $V=U_{1}\oplus U_{2}\oplus\cdots\oplus U_{m}$
be the canonical decomposition of $V$. If $\pi=(k_{11},\ldots,k_{1n},\ldots,k_{m1},\ldots,k_{mn})$ and $T:V\rightarrow V$ is a symmetry
such that $T(0)=0$, then for each index $1\leq i\leq m$ there corresponds
another index $1\leq j\leq m$ such that $$T(U_{i})=U_{j}$$ and $$k_{il}=\dim(V_{il})=\dim(V_{jl})=k_{jl}$$ for all $1 \leq l \leq n$.
\end{lemma}

\begin{proof}
In the following we denote the subspace $V_{i1} \oplus V_{i2} \oplus \ldots \oplus V_{ik}$ by $U_{ik}$. We begin by showing that for each index $1\leq i\leq m$ there corresponds another index $1\leq j\leq m$ such that $T(U_{i1})=U_{j1}$ and $k_{i1}=k_{j1}$.

Let $v_i \in U_{i1}$, $v_i \neq 0$. Since $$d_{(P,\pi)}(T(v_i),0)=d_{(P,\pi)}(v_i,0)=1,$$ then $T(v_i)$ is a vector of $(P,\pi)$-weight $1$. It follows that $T(v_i) \in U_{j1}$ for some index $1\leq j\leq m$. If $v_{i}' \in U_{i1}$, $v_{i}' \neq v_i$ and $v_{i}' \neq 0$, then $T(v_{i}')=v_k$ for some $v_k \in U_{k1}$ with $v_k \neq 0$, but also $$d_{(P,\pi)}(T(v_i),T(v_i'))=d_{(P,\pi)}(v_i,v_i')=1.$$ If $k \neq j$, then $d_{(P,\pi)}(T(v_i),T(v_i'))=d_{(P,\pi)}(v_j,v_k)=2$. Hence $k=j$ and $T(U_{i1}) \subseteq U_{j1}$.

Now apply the same reasoning to $T^{-1}$. If $v_i \in U_{i1}$, $v_i \neq 0$, and $T(v_i)=v_j$ with $v_j \in U_{j1}$, then $T^{-1}(v_j) \in U_{i1}$ and therefore $T^{-1}(U_{j1}) \subseteq U_{i1}$. So that $U_{j1} \subseteq T(U_{i1})$. If follows that $T(U_{i1})=U_{j1}$.

We have that $k_{i1}=k_{j1}$ because $T$ is bijective.

We will prove now, by induction on $k$, that for each $s$ there exists an index $l$ such
that
\[
T(U_{sk})=U_{lk}
\]
and $k_{sj}=k_{lj}$ for all $1\leq j\leq k$ and for all $1\leq k\leq n$. We note that $U_{sn} = U_s$.

Without loss of generality, let us consider $s=1$, $P_{1}=\{(1,1),\ldots,
(1,n)\}$. Let $P_{l}$ be the chain that begins at $(l,1)$ such that $T(U_{11})=U_{l1}$ and suppose that $U_{1(k-1)}$ is taken
by $T$ onto $U_{l(k-1)}$ with $k_{1j}=k_{lj}$ for all $1\leq j\leq k-1$.

Let $v=v_{11}+\ldots+v_{1k}$, $v_{1i} \in V_{1i}$, and let $%
T(v)=u_{1}+\ldots+u_{m}$, $u_{i}\in U_{i}$. Since $T(0)=0$,
\[
\omega_{(P,\pi)}(v)=\omega_{(P,\pi)}(T(v))=\omega_{(P,\pi)}(u_{1})+\ldots+\omega_{(P,\pi)}(u_{m}).
\]
We will use this to show that $T(v)=u_{l}$. First suppose that $u_{l}=0$. In
this case, $\omega_{(P,\pi)}(v)=\sum_{j\neq l}\omega_{(P,\pi)}(u_{j})$ and therefore, if $u_{11} \in U_{11}$, $u_{11} \neq 0$, with $T(u_{11})=u_{l1}$,
\[
k=d_{(P,\pi)}(u_{11},v)=d_{(P,\pi)}(T(u_{11}),T(v))=\sum_{j\neq
l}\omega_{(P,\pi)}(u_{j})+\omega_{(P,\pi)}(u_{l1})=k+1,
\]
a contradiction. Hence $u_{l}\neq0$.

Let $u_{l}=u_{l1}+\ldots+u_{lt}$, $u_{li} \in V_{li}$, and suppose now there
is another summand $u_{i}\neq0$. Then $k = \displaystyle\sum_{j} \omega
_{(P,\pi)}(u_{j}) > \omega_{(P,\pi)}(u_{l})$ and therefore $t<k$. By the induction
hypothesis, $T^{-1}(u_{l})$ is a vector in $V_{1(k-1)}$ with $\omega_{(P,\pi)}(T^{-1}(u_{l}))<k.$ Hence
\[
k=d_{(P,\pi)}(T^{-1}(u_{l}),v)=d_{(P,\pi)}(u_{l},T(v))=\sum_{j\neq
l}\omega_{(P,\pi)}(u_{j})<k,
\]
again a contradiction. Hence, $T(v) \in U_{lk}$. It follows from the
induction hypothesis and from the fact that $T$ is a weight-preserving
bijection that
\[
T(v_{11}+\ldots+v_{1k})=u_{l1}+\ldots+u_{lk}
\]
where $v_{1k} \neq 0$ implies $u_{lk} \neq 0$. Therefore $T(U_{1k}) =
U_{lk}$. Since $k_{1j}=k_{lj}$ for all $1\leq j\leq k-1$ and $T$ is a bijection, it follows that $k_{1k}=k_{lk}$. Hence $T(U_{1})=U_{l}$ with $k_{1j}=k_{lj}$ for all $1\leq j\leq n$.
\end{proof}

We recall that we defined an action of the group $S_{\pi}$ of the admissible
permutations of $S_m$ on the canonical decomposition $U_{1}\oplus U_{2}\oplus\cdots\oplus U_{m}$ of $V$ by
\[
T_{\sigma}(v):=v_{\sigma(1)}+v_{\sigma(2)}+\ldots+v_{\sigma(m)}
\]
and that we defined an action of $\prod_{i=1}^m G_{i,\pi_{i},n}$ on $V$ by
\[
(g_{1},g_{2},\ldots,g_{m})(v_{1}+v_{2}+\ldots+v_{m})=g_{1}(v_{1})+\ldots
+g_{m}(v_{m}).
\]

\begin{lemma}
\label{lemma}Let $(V,d_{(P,\pi)})$ be the $(m,n,\pi)$-ordered Hamming block space. Each symmetry of $V$ that preserves the origin is a
product $T_{\sigma} \circ g$, with $\sigma$ in $S_{\pi}$ and $g$ in $\prod_{i=1}^m G_{i,\pi_{i},n}$.
\end{lemma}

\begin{proof}
Let $T$ be a symmetry of $V$, $T(0)=0$. By the previous result, for each $%
1\leq i \leq m$ there is a $\sigma(i)$ such that $T(U_{i})=U_{\sigma(i)}$ with $k_{il}=k_{\sigma(i)l}$ for all $1\leq l \leq n$. Since $T$ is a bijection, it follows that the map $i\mapsto \sigma(i)$ is an admissible permutation
of the set $\{1,\ldots,m\}$. We define $T_{\sigma}: V \rightarrow V $ by
\[
T_{\sigma}(v):=v_{\sigma(1)}+v_{\sigma(2)}+\ldots+v_{\sigma(m)}
\]
and then $T=T_{\sigma}(T_{\sigma}^{-1}T)$, where $\sigma\in S_{\pi} $. Let $%
g=(T_{\sigma} ^{-1}T)$. Clearly $g(U_{i})=U_{i}$, and $g|_{U_{i}}$ is a
symmetry of $U_{i}$. Defining $g_{i}:=g|_{U_{i}}$ we have that $g=(g_{1},\ldots,g_{m})$ and hence $g \in \prod_{i=1}^m G_{i,\pi_{i},n}$.
\end{proof}

\begin{theorem}
\label{Theorem RT}\label{Theorem}Let $(V,d_{(P,\pi)})$ be the $(m,n,\pi)$-ordered Hamming block space. The group of symmetries of
$V$ is isomorphic to $$\left(\prod_{i=1}^m G_{i,\pi_{i},n}\right) \mathbb{o} S_{\pi}.$$
\end{theorem}

\begin{proof}
As before, let $G_{(m,n,\pi)}$ be the group of symmetries of $V$ generated by
the action of $\prod_{i=1}^m G_{i,\pi_{i},n}$ and $S_{\pi}$. Let $T$ be a symmetry
of $V $ and let $v=T(0)$. The translation $S_{-v}(u):=u-v$ is clearly a
symmetry of $V$ and $\left( S_{-v}\circ T\right) (0)=S_{-v}(v)=0$ is a
symmetry that fixes the origin. Hence, by the previous lemma, $S_{-v}\circ
T\in G_{(m,n,\pi)}$. Consider the canonical decomposition of $v$ on the chain spaces, $%
v=v_{1}+\ldots +v_{m},v_{i}\in U_{i}$. We see that the restriction of $S_{v}
$ to $U_{i}$ is the translation by $v_{i}$, hence a symmetry of $U_{i}$. It
follows that $S_{v}\in\prod_{i=1}^m G_{i,\pi_{i},n}\subset G_{(m,n,\pi)}$ and
hence, that $T=S_{v}\circ(S_{-v}\circ T)$ is in $G_{(m,n,\pi)}$ and we
conclude that $G_{(m,n,\pi)}$ is the symmetry group of $V$. By Proposition
\ref{Gmn}, $G_{(m,n,\pi)}$ is isomorphic to $\left(\prod_{i=1}^m G_{i,\pi_{i},n}\right) \mathbb{o} S_{\pi}$.
\end{proof}

If $n=1$ ($P$ is an antichain) and $\pi=(k_1,k_2,\ldots,k_m)$, where $$k_1=\ldots=k_{m_1}=l_1, \ldots, k_{m_1+\ldots +m_{l-1}+1}=\ldots =k_{m_1+\ldots +m_{l}}=l_r $$ with $l_1>\ldots>l_r$, we have that  $G_{i,(k_i),1} = S_{q^{k_i}}$, $1 \leq i \leq m$, and $S_{\pi}=S_{m_1} \times \ldots \times S_{m_l}$ ($S_{\pi}$ only permutes those blocks with same dimensions). Therefore:

\begin{corollary}
If $P$ is an antichain, then $$\textit{Symm}(V,d_{(P,\pi)})=\left( \prod_{i=1}^{m}S_{q^{k_i}} \right) \mathbb{o} \left( \prod_{i=1}^{l} S_{m_i} \right).$$
\end{corollary}

When $n=1$ and $\pi = (1,1,\ldots,1)$, the $(P,\pi)$-weight is the usual Hamming weight on $\mathbb{F}_q^{m}$. In this case each $G_{i,(1),1}$ in above corollary is equal to $S_q$ and every permutation in $S_m$ is also admissible. Thus we reobtain the symmetry groups of Hamming space from our previous calculations:

\begin{corollary}
Let $d_H$ be the Hamming metric over $\mathbb{F}_q^{m}$. The symmetry group of $(\mathbb{F}_q^{m},d_H)$ is isomorphic to $S_q^m \mathbb{o} S_m$.
\end{corollary}

If $\pi = (1,1,\ldots,1)$ also every permutation in $S_m$ is admissible. Hence (see \cite{Alves2}, Theorem 4.1):

\begin{corollary}
Let $P=([mn],\leq )$ be the ordered Hamming space. Let $V=\mathbb{F}_{q}^{mn}$ be the vector space endowed with the poset metric $d_P$ induced by the
poset $P$. Then the group of simmetries $\textit{Symm}(V,d_{P})$ is isomorphic to the semi-direct product $$\left(G_n \right)^m \mathbb{o} S_m$$ where $G_n := 
(S_{q})^{q^{n-1}} \mathbb{o} ( \ldots ((S_{q})^{q} \mathbb{o} S_{q}) \ldots )$. In particular, $$\vert \textit{Symm}(V,d_{P}) \vert = (q!)^{m \cdot \frac{q^n-1}{q-1}+m}\cdot m!.$$
\end{corollary}

\section{Automorphisms}

The group of automorphisms of $\left( V,d_{(P,\pi)}\right) $ is easily
deduced from the results above. Let $T=T_{\sigma} \circ g$ be a symmetry. Since $%
T_{\sigma}$ is linear, the linearity of $T$ is a matter of whether $g$ is
linear or not. Now, if $g=\left( g_{1},g_{2},\ldots ,g_{m}\right) $ is
linear, then each component $g_{i}$ must also be linear; since each $g_{i}$
is bijective, $g_{i}$ is in the group $Aut\left( U_{i}\right) $ of linear
automorphisms of $U_{i}$. Therefore $g\in \prod_{i=1}^{m}Aut\left(
U_{i}\right) $. On the other hand, any element of this group is a linear
symmetry. Hence:

\begin{theorem}
The automorphism group $Aut\left( V,d_{(P,\pi)}\right) $ of $\left( V,d_{(P,\pi)}\right)$ is isomorphic to $$\left(\prod_{i=1}^{m}Aut\left( V_{i}\right)\right) \mathbb{o} S_{\pi }.$$
\end{theorem}

\begin{corollary}
Let $n=1$ and $\pi =\left( k_{1},k_{2},\ldots ,k_{m}\right) $ be a partition of $N$.
If
\[k_{1}=\ldots =k_{m_{1}}=l_{1},\ldots ,k_{m_{1}+\ldots +m_{l-1}+1}=\ldots
=k_{n}=l_{r}\] with $l_{1}>l_{2}>\ldots >l_{r}$, then
\[
\left| Aut\left( \Bbb{F}_{q}^{N},d_{(P,\pi)}\right) \right| = \left( \prod_{i=1}^{m}\left(
q^{k_{i}}-1\right) \left( q^{k_{i}}-q\right) \ldots \left(
q^{k_{i}}-q^{k_{i}-1}\right) \right) \cdot \left(
\prod_{j=1}^{l}m_{j}!\right) \text{.}
\]
\end{corollary}

\begin{proof}
Note initially that there is a bijection from $Aut\left(
V_{i}\right) $ and the family of all ordered bases of $V_{i}$: let
$\left( e_{1},e_{2},\ldots ,e_{k_{i}}\right) $ be an ordered basis
of $V_{i}$; if $T\in Aut\left( V_{i}\right) $, then $\left(
T\left( e_{1}\right) ,T\left( e_{2}\right) ,\ldots ,T\left(
e_{k_{i}}\right) \right) $ is an ordered basis of $V_{i}$;
if  $\left( v_{1},v_{2},\ldots ,v_{k_{i}}\right) $ is an ordered basis of $%
V_{i}$ then there exist a unique automorphism $T$ with $T\left(
e_{j}\right) =v_{j}$ for all $j\in \left\{ 1,2,\ldots
,k_{i}\right\} $. Since the number of ordered basis of $V_{i}$
equal
\[
\left( q^{k_{i}}-1\right) \left( q^{k_{i}}-q\right) \ldots \left(
q^{k_{i}}-q^{k_{i}-1}\right)
\]
follows that $\left| Aut\left( V_{i}\right) \right| =\left(
q^{k_{i}}-1\right) \left( q^{k_{i}}-q\right) \ldots \left(
q^{k_{i}}-q^{k_{i}-1}\right) $. From above theorem
\[
\left| Aut\left( \Bbb{F}_{q}^{N},d_{(P,\pi)}\right) \right| = \left(
\prod_{i=1}^{m}\left| Aut\left( V_{i}\right) \right| \right) \cdot \left| S_{\pi }\right| \text{.}
\]
Since $\left| S_{\pi }\right| =\prod_{j=1}^{l}m_{j}!$ the corollary follows.
\end{proof}

Restricting to the Hamming case again, $Aut\left( V_{i}\right) =Aut\left(
\Bbb{F}_{q}\right) \simeq \Bbb{F}_{q}^{*}$ and $S_{\pi }=S_{m}$. Hence:

\begin{corollary}
The automorphism group of $\left( \Bbb{F}_{q}^{m},d_{H}\right) $ is $%
\left( \Bbb{F}_{q}^{*}\right) ^{m} \mathbb{o} S_{m}$.
\end{corollary}

\end{document}